\documentclass[11pt,amssymb,amsfont,a4paper]{article}
\usepackage{latexsym,amssymb,amsmath,amsthm,color}
\usepackage{geometry}
\usepackage{url}
\usepackage{graphicx}
\usepackage[utf8]{inputenc}
\usepackage[numbers]{natbib}
\usepackage{authblk}
\usepackage{tikz}

\theoremstyle{definition}
\newtheorem{definition}{Definition}
\newtheorem{example}[definition]{Example}

\theoremstyle{plain}
\newtheorem{theorem}{Theorem}
\newtheorem{proposition}[definition]{Proposition}
\newtheorem{lemma}[definition]{Lemma}
\newtheorem{remark}[definition]{Remark}
\newtheorem{corollary}[definition]{Corollary}

\title{Skew and linearized Reed-Solomon codes and maximum sum rank distance codes over any division ring}
\author[1,2]{Umberto Mart{\'i}nez-Pe\~{n}as \thanks{umberto@math.aau.dk}}
\affil[1]{Dept.\ of Mathematical Sciences, Aalborg University, Denmark}
\affil[2]{Dept.\ of Electrical \& Computer Engineering,
University of Toronto, Canada}

\date{}

\begin{document}

\maketitle

\begin{abstract}
Reed-Solomon codes and Gabidulin codes have maximum Hamming distance and maximum rank distance, respectively. A general construction using skew polynomials, called skew Reed-Solomon codes, has already been introduced in the literature. In this work, we introduce a linearized version of such codes, called linearized Reed-Solomon codes. We prove that they have maximum sum-rank distance. Such distance is of interest in multishot network coding or in singleshot multi-network coding. To prove our result, we introduce new metrics defined by skew polynomials, which we call skew metrics, we prove that skew Reed-Solomon codes have maximum skew distance, and then we translate this scenario to linearized Reed-Solomon codes and the sum-rank metric. The theories of Reed-Solomon codes and Gabidulin codes are particular cases of our theory, and the sum-rank metric extends both the Hamming and rank metrics. We develop our theory over any division ring (commutative or non-commutative field). We also consider non-zero derivations, which give new maximum rank distance codes over infinite fields not considered before.


\end{abstract}

\section{Introduction} \label{sec intro}

The Hamming metric has always played a central role in error correction. \textit{Reed-Solomon codes} \cite{reed-solomon} were the first general linear codes to achieve maximum Hamming distance. A similar construction called \textit{Gabidulin codes} was introduced in \cite{gabidulin, new-construction} using linearized polynomials over finite fields \cite[Chapter 3]{lidl}. Such codes have maximum rank distance, which have made them gain attention in connection with linear network coding \cite{on-metrics, silva-universal}. 

Both types of codes are constructed by using skew polynomials in a way. Skew polynomial rings are the most general polynomial rings where multiplication is additive on degrees, and were introduced in \cite{ore}. The theory of \textit{evaluation} and \textit{interpolation} for such skew polynomials was developed in \cite{lam, algebraic-conjugacy, lam-leroy}. A general construction of codes based on evaluations of skew polynomials, called \textit{skew Reed-Solomon codes}, was introduced in \cite{skew-evaluation1}, and further studied in \cite{skew-evaluation2}. It is proven in \cite{skew-evaluation1} that all such codes are maximum Hamming distance (MDS) codes, and they are maximum rank distance (MRD) when defined on one conjugacy class \cite{algebraic-conjugacy}. A particular case of skew Reed-Solomon codes with MRD components is given in \cite{skew-evaluation2} for larger lengths but low dimensions.

The connection between skew Reed-Solomon codes on one conjugacy class and the \textit{linearized} structure of Gabidulin codes was given in \cite[Section 4]{skew-evaluation1} and in \cite{skew-evaluation2}. However, the \textit{linearized} structure of skew Reed-Solomon codes is unknown in general. In this work, we introduce a new family of codes, called \textit{linearized Reed-Solomon codes}, and show that they indeed are the linearized version of skew Reed-Solomon codes in all cases.

Next we show that all linearized Reed-Solomon codes have maximum \textit{sum-rank distance}. This distance has been introduced in \cite{multishot} in the context of convolutional codes for multishot network coding. Optimal constructions for modified sum-rank metrics or for the sum-rank metric of convolutional codes have been given in \cite{mahmood-convolutional, mrd-convolutional, wachter, wachter-convolutional}. 

To the best of our knowledge, the introduced linearized Reed-Solomon codes are the first general linear \textit{block codes} with maximum sum-rank distance, except for the (extreme) particular cases of Reed-Solomon and Gabidulin codes. The use of maximum sum-rank distance block codes is of interest in multishot network coding when the number of shots is low and known beforehand. It can also be of interest in singleshot multi-network coding, where errors and erasures may spread over several networks that could even have different numbers of outgoing links from the source. Moreover, the sum-rank metric is a hybrid metric that gives a common theoretical framework for both the Hamming and rank metrics\footnote{Interestingly, the term ``sum'' carries the Hamming part of the metric, and the term ``rank'' carries the rank part. We do not know if this popular terminology was intentional in this sense.}.

We prove that linearized Reed-Solomon codes have maximum sum-rank distance as follows. We introduce a new family of metrics defined by skew polynomials, called \textit{skew metrics}, and show that all skew Reed-Solomon codes have maximum skew distance. Based on some results in \cite{lam, algebraic-conjugacy, lam-leroy, hilbert90, leroy-pol}, we then prove that sum-rank metrics are the linearized versions of skew metrics, and linearized Reed-Solomon codes are the linearized versions of skew Reed-Solomon codes, which concludes our proof. A shorter proof is possible, although based essentially on the same algebraic machinery \footnote{One can define the right notion of \textit{erasures} characterizing the sum-rank metric, and then use the bound by degrees on the \textit{sum-dimensions} of zero sets of linear operator polynomials \cite[Theorem 2.1]{leroy-noncommutative}.}. However, such proof obscures the connection between skew and linearized Reed-Solomon codes, which has not been established yet.

We develop our theory over any division ring. Rank-metric codes over the complex field have been studied in \cite{augot} with a view towards space-time coding \cite{space-time}. Codes based on skew polynomials over transcendental extensions of finite fields have been studied in \cite{augot-function}, and a connection with cyclic convolutional codes has been established in \cite{torrecillas}, which includes some skew Reed-Solomon codes. A general study of Gabidulin and rank-metric codes over any field is given in \cite{augot-extended}. On the other hand, developing the theory in the general case eliminates some anomalies that are characteristic of finite fields. The reader only interested in fields and not in division rings may just omit the words ``left'' and ``right'' when considering the multiplicative structure of fields and their vector spaces. 

The organization is as follows: In Section \ref{sec skew metrics}, we recall the main definitions and results concerning skew polynomial rings \cite{lam, algebraic-conjugacy, lam-leroy, ore}. We then introduce the new family of skew metrics (Definition \ref{def skew metrics}) and we show that all skew Reed-Solomon codes are maximum skew distance codes (Theorem \ref{th maximum skew distance codes}). In Section \ref{sec linearizing skew metrics}, we recall the concept of linear operator polynomials and its connection with skew polynomial evaluation \cite{leroy-pol}. We then introduce linearized Reed-Solomon codes (Definition \ref{def linearized RS codes}). Next we show that the linearized version of skew metrics and skew Reed-Solomon codes are sum-rank metrics and linearized Reed-Solomon codes (Theorems \ref{th big commutative diagram} and \ref{th both weights coincide} and Proposition \ref{prop indeed linearized version}), respectively. We deduce then that the latter codes are maximum sum-rank distance (Theorem \ref{th max sum rank distance}). We conclude in Section \ref{sec particular cases} with explicit descriptions of the studied objects for (commutative) fields. These include finite fields and a field where linearized Reed-Solomon codes can only be constructed by using derivations instead of endomorphisms.

\section{Skew metrics and skew Reed-Solomon codes} \label{sec skew metrics}

In this section, we define a new family of skew polynomial weights and their corresponding metrics, which we call \textit{skew metrics}. We then show a tight connection between the new metrics and the Hamming metric. We conclude by showing that skew Reed-Solomon codes are maximum skew distance codes.

\subsection{Skew polynomials over division rings} \label{subsec skew polynomials division rings}

In this subsection, we will collect some definitions and results from the literature concerning skew polynomial rings \cite{ore} and their evaluation maps over division rings \cite{lam, algebraic-conjugacy, lam-leroy}. 

Fix a division ring $ \mathbb{F} $ from now on and denote by $ \mathbb{N} $ the set of natural numbers including $ 0 $. Let $ \mathcal{R} $ be the left vector space over $ \mathbb{F} $ with basis $ \{ x^i \mid i \in \mathbb{N} \} $, where we denote $ 1 = x^0 $. Define then the degree of a non-zero element $ F = \sum_{i \in \mathbb{N}} F_i x^i \in \mathcal{R} $, where $ F_i \in \mathbb{F} $ for all $ i \in \mathbb{N} $, as the maximum $ i $ such that $ F_i \neq 0 $, and denote it by $ \deg(F) $. We also define $ \deg(0) = \infty $.

It is shown in \cite{ore} that a product in $ \mathcal{R} $ turns it into a (non-commutative) ring with multiplicative identity $ 1 $, where $ x^i x^j = x^{i+j} $, for all $ i,j \in \mathbb{N} $, and $ \deg(FG) = \deg(F) + \deg(G) $ for all $ F,G \in \mathcal{R} $, if, and only if, there exist $ \sigma, \delta : \mathbb{F} \longrightarrow \mathbb{F} $ such that
\begin{equation}
xa = \sigma(a) x + \delta(a),
\label{eq product over constants and variables}
\end{equation}
for all $ a \in \mathbb{F} $, where $ \sigma : \mathbb{F} \longrightarrow \mathbb{F} $ is a (ring) endomorphism and $ \delta : \mathbb{F} \longrightarrow \mathbb{F} $ is a $ \sigma $-derivation: That is, $ \delta $ is additive and for all $ a,b \in \mathbb{F} $, it holds that
$$ \delta(ab) = \sigma(a)\delta(b) + \delta(a)b. $$

For each such pair $ (\sigma,\delta) $, we use the notation $ \mathcal{R} = \mathbb{F}[x; \sigma, \delta] $ when considering in $ \mathcal{R} $ the product given by (\ref{eq product over constants and variables}), and we call $ \mathbb{F}[x; \sigma, \delta] $ the \textit{skew polynomial ring} over $ \mathbb{F} $ with endomorphism $ \sigma $ and derivation $ \delta $. The conventional polynomial ring $ \mathbb{F}[x] $ is obtained by choosing $ \sigma = {\rm Id} $ and $ \delta = 0 $.

As shown in \cite{ore}, the rings $ \mathbb{F}[x; \sigma, \delta] $ are left and right Euclidean domains. Hence, we may give a natural definition of evaluation of skew polynomials by forcing a ``Remainder Theorem''. This is the approach in \cite{lam, lam-leroy}:

\begin{definition} [\textbf{Evaluation \cite{lam, lam-leroy}}]
Given $ F \in \mathbb{F}[x; \sigma, \delta] $, we define its evaluation over a point $ a \in \mathbb{F} $ as the unique element $ F(a) \in \mathbb{F} $ such that there exists $ G \in \mathbb{F}[x; \sigma, \delta] $ with
$$ F = G (x-a) + F(a). $$
Given a subset $ \Omega \subseteq \mathbb{F} $, we denote by $ \mathbb{F}^\Omega $ the family of (arbitrary) functions $ f : \Omega \longrightarrow \mathbb{F} $. We then define the evaluation map over $ \Omega $ as the left linear map
\begin{equation}
E^{\sigma,\delta}_\Omega : \mathbb{F}[x; \sigma, \delta] \longrightarrow \mathbb{F}^\Omega,
\label{def evaluation map}
\end{equation}
where $ f = E^{\sigma,\delta}_\Omega (F) \in \mathbb{F}^\Omega $ is given by $ f(a) = F(a) $, for all $ a \in \Omega $ and for $ F \in \mathbb{F}[x; \sigma, \delta] $. We will use the notation $ E_\Omega $ whenever $ \sigma $ and $ \delta $ are understood from the context.
\end{definition}

The structure of sets of zeros of skew polynomials was extensively studied in \cite{lam, algebraic-conjugacy, lam-leroy} and will be crucial for our purposes. In particular, the main result is that of the existence and uniqueness of Lagrange interpolating polynomials (Lemma \ref{lemma lagrange interpolation}). 

\begin{definition} [\textbf{Zeros of skew polynomials}]
Given a set $ A \subseteq \mathbb{F}[x;\sigma,\delta] $, we define its zero set as
$$ Z(A) = \{ a \in \mathbb{F} \mid F(a) = 0, \forall F \in A \}. $$
Given a subset $ \Omega \subseteq \mathbb{F} $, we define its associated ideal as
$$ I(\Omega) = \{ F \in \mathbb{F}[x;\sigma, \delta] \mid F(a) = 0, \forall a \in \Omega \}. $$
\end{definition}

Observe that $ I(\Omega) $ indeed is a left ideal in $ \mathbb{F}[x;\sigma,\delta] $, for any subset $ \Omega \subseteq \mathbb{F} $. Since $ \mathbb{F}[x;\sigma,\delta] $ is a right Euclidean domain, there exists a unique monic skew polynomial $ F_\Omega \in I(\Omega) $ of minimal degree among those in $ I(\Omega) $, which in turn generates $ I(\Omega) $ as left ideal. We will call such skew polynomial the \textit{minimal skew polynomial} of $ \Omega $. This motivates the concepts of \textit{P-closed sets}, \textit{P-independence} and \textit{P-bases}, which we take from \cite[Section 4]{algebraic-conjugacy} (see also \cite{lam}):

\begin{definition} [\textbf{P-bases \cite{lam, algebraic-conjugacy}}]
Given a subset $ \Omega \subseteq \mathbb{F} $, we define its P-closure as $ \overline{\Omega} = Z(I(\Omega)) = Z(F_\Omega) $, and we say that it is P-closed if $ \overline{\Omega} = \Omega $.

Given a P-closed set $ \Omega \subseteq \mathbb{F} $, we say that $ \mathcal{G} \subseteq \Omega $ generates it if $ \overline{\mathcal{G}} = \Omega $ (equivalently, $ F_\mathcal{G} = F_\Omega $), and it is called a set of P-generators for $ \Omega $. We say that $ \Omega $ is finitely generated if it has a finite set of P-generators.

We say that $ a \in \mathbb{F} $ is P-independent from $ \Omega \subseteq \mathbb{F} $ if it does not belong to $ \overline{\Omega} $ (equivalently, $ F_{\Omega \cup \{ a \}} \neq F_\Omega $). A set $ \Omega \subseteq \mathbb{F} $ is called P-independent if every $ a \in \Omega $ is P-independent from $ \Omega \setminus \{ a \} $. 

Given a P-closed set $ \Omega \subseteq \mathbb{F} $, we say that a subset $ \mathcal{B} \subseteq \Omega $ is a P-basis of $ \Omega $ if it is P-independent and a set of P-generators of $ \Omega $.
\end{definition}

The following results are given in \cite{lam, algebraic-conjugacy}:

\begin{lemma}[\textbf{\cite{lam}}]
Given a finite set $ \Omega \subseteq \mathbb{F} $, it holds that $ \deg(F_\Omega) \leq \# \Omega $. Furthermore, $ \Omega $ is P-independent if, and only if, $ \# \Omega = \deg(F_\Omega) $.
\end{lemma}

\begin{corollary} [\textbf{\cite{lam, algebraic-conjugacy}}] \label{cor charact P-bases}
Given a finitely generated P-closed set $ \Omega \subseteq \mathbb{F} $ and a finite subset $ \mathcal{B} \subseteq \Omega $, the following are equivalent:
\begin{enumerate}
\item
$ \mathcal{B} $ is a P-basis of $ \Omega $.
\item
$ \mathcal{B} $ is a maximal P-independent subset of $ \Omega $.
\item
$ \mathcal{B} $ is a minimal set of P-generators of $ \Omega $.
\end{enumerate}
In particular, $ \Omega $ admits at least one finite P-basis.
\end{corollary}

\begin{corollary} [\textbf{\cite{lam, algebraic-conjugacy}}] \label{cor any two basis same size}
Given a finitely generated P-closed set $ \Omega \subseteq \mathbb{F} $, any two of its P-bases are finite and have the same number of elements, which is $ \deg(F_\Omega) $.
\end{corollary}

This motivates the following definition:

\begin{definition} [\textbf{Ranks \cite{lam, algebraic-conjugacy}}]
Given a finitely generated P-closed set $ \Omega \subseteq \mathbb{F} $, we define its rank as 
$$ {\rm Rk}(\Omega) = \deg(F_\Omega), $$
which is the size of any P-basis of $ \Omega $.
\end{definition}

From now on, all P-closed sets in $ \mathbb{F} $ are assumed to be finitely generated. Observe that this is always the case unless $ \Omega = \mathbb{F} $ and $ \mathbb{F} $ is infinite.

We have now arrived at the main result of this section, which is the existence and uniqueness of Lagrange interpolating skew polynomials. Although we use our own notation, the next lemma follows directly from \cite[Theorem 8]{lam}, which states that a skew Vandermonde matrix is invertible if, and only if, it is defined over a P-independent set. 

\begin{lemma}[\textbf{Lagrange interpolation \cite{lam}}] \label{lemma lagrange interpolation}
Let $ \Omega \subseteq \mathbb{F} $ be a P-closed set with P-basis $ \mathcal{B} = \{ b_1, b_2, \ldots, b_n \} $. For every $ a_1, a_2, \ldots, a_n \in \mathbb{F} $, there exists a unique $ F \in \mathbb{F}[x; \sigma, \delta] $ such that $ \deg(F) < n $ and $ F(b_i) = a_i $, for $ i = 1,2, \ldots, n $.
\end{lemma}

\subsection{Skew polynomial metrics} \label{subsec skew pol metrics}

Fix a P-closed set $ \Omega \subseteq \mathbb{F} $ with (finite) P-basis $ \mathcal{B} $, and write $ n = \# \mathcal{B} = {\rm Rk}(\Omega) $. Observe that $ n \leq {\rm Rk}(\mathbb{F}) $. Denote by $ \mathbb{F}[x; \sigma, \delta]_n $ the $ n $-dimensional left vector space of skew polynomials of degree less than $ n $. It follows from Lemma \ref{lemma lagrange interpolation} that the evaluation map (\ref{def evaluation map}) restricted to $ \mathbb{F}[x; \sigma, \delta]_n $, that is
$$ E_\mathcal{B} : \mathbb{F}[x; \sigma, \delta]_n \longrightarrow \mathbb{F}^\mathcal{B}, $$
is a left vector space isomorphism. We may thus define the mentioned skew polynomial weights in the space $ \mathbb{F}[x; \sigma, \delta]_n $ and then extend them to $ \mathbb{F}^\mathcal{B} $:

\begin{definition} [\textbf{Skew weights}]
For $ F \in \mathbb{F}[x; \sigma, \delta]_n $, we define its $ (\sigma,\delta) $-skew polynomial weight, or just skew weight for simplicity, over $ \Omega $ as
$$ {\rm wt}^{\sigma, \delta}_\Omega(F) = {\rm Rk}(\Omega) - {\rm Rk}(Z_\Omega(F)) = n - {\rm Rk}(Z_\Omega(F)), $$
where $ Z_\Omega(F) = Z(F) \cap \Omega = Z(\{ F, F_\Omega \}) $ is the P-closed set of zeros of $ F $ in $ \Omega $. Observe that $ 0 \leq {\rm wt}^{\sigma, \delta}_\Omega(F) \leq n $.

Now, for an arbitrary vector $ f \in \mathbb{F}^\mathcal{B} $, there exist a unique $ F \in \mathbb{F}[x; \sigma, \delta]_n $ such that $ f = E_\mathcal{B} (F) $ by Lemma \ref{lemma lagrange interpolation}. We then define 
$$ {\rm wt}^{\sigma, \delta}_\mathcal{B}(f) = {\rm wt}^{\sigma, \delta}_\Omega(F). $$

We will write $ {\rm wt}_\Omega $ and $ {\rm wt}_\mathcal{B} $ whenever $ \sigma $ and $ \delta $ are known from the context.
\end{definition}

The following properties will allow us to define the desired metrics:

\begin{proposition} \label{prop properties of weights}
Let $ F,G \in \mathbb{F}[x; \sigma, \delta]_n $ and $ a \in \mathbb{F}^* $. The following properties hold:
\begin{enumerate}
\item
$ {\rm wt}_\Omega(F) = 0 $ if, and only if, $ E_\mathcal{B}(F) = 0 $, which is equivalent to $ F = 0 $.
\item
$ {\rm wt}_\Omega(F + G) \leq {\rm wt}_\Omega(F) + {\rm wt}_\Omega(G) $.
\item
$ {\rm wt}_\Omega(aF) = {\rm wt}_\Omega(F) $.
\end{enumerate}
The same properties hold for $ {\rm wt}_\mathcal{B} $ in $ \mathbb{F}^\mathcal{B} $ by definition.
\end{proposition}
\begin{proof}
We prove each item separately:
\begin{enumerate}
\item
We have the following chain of equivalences:
$$ {\rm wt}_\Omega(F) = 0 \Longleftrightarrow {\rm Rk}(\Omega) = {\rm Rk}(Z_\Omega(F)) \Longleftrightarrow Z_\Omega(F) = \Omega \Longleftrightarrow E_\Omega(F) = 0. $$
Finally, $ E_\Omega(F) = 0 $ is equivalent to $ E_\mathcal{B}(F) = 0 $ and to $ F = 0 $ by Lemma \ref{lemma lagrange interpolation}.
\item
It holds that
$$ Z_\Omega(F) \cap Z_\Omega(G) \subseteq Z_\Omega(F+G). $$
Denote $ \Omega_1 = Z_\Omega(F) $ and $ \Omega_2 = Z_\Omega(G) $, which are P-closed. The result \cite[Theorem 7]{matroidal} says that
$$ \deg(F_{\Omega_1 \cup \Omega_2}) + \deg(F_{\Omega_1 \cap \Omega_2}) = \deg(F_{\Omega_1}) + \deg(F_{\Omega_2}). $$
We conclude that
$$ n - {\rm Rk}(Z_\Omega(F+G)) \leq n - {\rm Rk}(\Omega_1 \cap \Omega_2) $$
$$ \leq (n - {\rm Rk}(\Omega_1)) + (n - {\rm Rk}(\Omega_2)) $$
since $ \deg(F_{\Omega_1 \cup \Omega_2}) \leq n $ because $ \Omega_1 \cup \Omega_2 \subseteq \Omega $ and $ n = {\rm Rk}(\Omega) $. Thus the result follows.
\item
Trivial from $ Z_\Omega(aF) = Z_\Omega(F) $, which holds since $ a \in \mathbb{F}^* $.
\end{enumerate}
\end{proof}

Therefore, the following functions are indeed metrics:

\begin{definition} [\textbf{Skew metrics}] \label{def skew metrics}
We define the $ (\sigma,\delta) $-skew polynomial metric, or just skew metric for simplicity, over $ \Omega $ as the function $ {\rm d}_\Omega^{\sigma, \delta} : \mathbb{F}[x; \sigma, \delta]_n^2 \longrightarrow \mathbb{N} $ given by
$$ {\rm d}_\Omega^{\sigma, \delta}(F,G) = {\rm wt}^{\sigma, \delta}_\Omega (F - G), $$
for all $ F,G \in \mathbb{F}[x; \sigma, \delta]_n $. By the left vector space isomorphism $ E_\mathcal{B} $, we define the corresponding metric $ {\rm d}_\mathcal{B}^{\sigma, \delta} : (\mathbb{F}^\mathcal{B})^2 \longrightarrow \mathbb{N} $ by 
$$ {\rm d}_\mathcal{B}^{\sigma, \delta}(f,g) = {\rm d}_\Omega^{\sigma, \delta}(F,G), $$
where $ f = E_\mathcal{B} (F) $ and $ g = E_\mathcal{B} (G) $, for  $ F,G \in \mathbb{F}[x; \sigma, \delta]_n $.

We will write $ {\rm d}_\Omega $ and $ {\rm d}_\mathcal{B} $ whenever $ \sigma $ and $ \delta $ are known from the context.
\end{definition}

We next observe that changing P-bases preserves the defined weights and metrics. This fact provides a family of left linear isometries for the new metrics:

\begin{definition} \label{def change of P-basis}
Given another P-basis $ \mathcal{A} $ of $ \Omega $, we define the change-of-P-basis map
\begin{equation*}
\begin{array}{rccc}
\pi_{\mathcal{B},\mathcal{A}} : & \mathbb{F}^\mathcal{B} & \longrightarrow & \mathbb{F}^\mathcal{A} \\
 & E_\mathcal{B}(F) & \mapsto & E_\mathcal{A}(F)
\end{array}
\end{equation*}
for all $ F \in \mathbb{F}[x; \sigma, \delta]_n $.
\end{definition}

Observe that such maps are well-defined by Lemma \ref{lemma lagrange interpolation}: Given $ F,G \in \mathbb{F}[x; \sigma, \delta]_n $ such that $ E_\mathcal{B}(F) = E_\mathcal{B}(G) $, it holds that $ F = G $ (Lemma \ref{lemma lagrange interpolation}), and hence $ E_\mathcal{A}(F) = E_\mathcal{A}(G) $. 

The following result thus follows from the definitions:

\begin{proposition}
Given another P-basis $ \mathcal{A} $ of $ \Omega $, the map $ \pi_{\mathcal{B},\mathcal{A}} : \mathbb{F}^\mathcal{B} \longrightarrow \mathbb{F}^\mathcal{A} $ is a left vector space isomorphism and
$$ {\rm wt}_\mathcal{B}(f) = {\rm wt}_\mathcal{A}(\pi_{\mathcal{B},\mathcal{A}}(f)), $$
for all $ f \in \mathbb{F}^\mathcal{B} $. That is, $ \pi_{\mathcal{B},\mathcal{A}} $ is a left linear isometry.
\end{proposition}

We conclude the section by giving a connection between the new metrics and the Hamming metric. Intuitively, this connection is just the observation that $ {\rm wt}_\Omega(F) $ does not depend on the chosen P-basis $ \mathcal{B} $, while at the same time we consider it in $ \mathbb{F}^\mathcal{B} $, which does depend on $ \mathcal{B} $, hence we may run over all P-bases of $ \Omega $. This fact has been observed for generalized rank and Hamming weights in \cite[Theorems 1 \& 6]{similarities}. We will also use this result to connect skew metrics with sum-rank metrics in Subsection \ref{subsec sum rank metric}.

For all $ f \in \mathbb{F}^{\mathcal{A}} $, we will denote its \textit{Hamming weight} by
$$ {\rm wt_H}(f) = \# \{ a \in \mathcal{A} \mid f(a) \neq 0 \}. $$

\begin{proposition} \label{prop connection with Hamming weights}
For any $ F \in \mathbb{F}[x; \sigma, \delta]_n $ and $ f = E_\mathcal{B}(F) \in \mathbb{F}^\mathcal{B} $, it holds that
$$ {\rm wt}_\Omega (F) = \min \{ {\rm wt_H}(E_\mathcal{A}(F)) \mid \mathcal{A} \textrm{ is a P-basis of } \Omega \} \leq {\rm wt_H}(E_\mathcal{B}(F)), $$
or equivalently, 
$$ {\rm wt}_\mathcal{B}(f) = \min \{ {\rm wt_H}(\pi_{\mathcal{B},\mathcal{A}}(f)) \mid \mathcal{A} \textrm{ is a P-basis of } \Omega \} \leq {\rm wt_H}(f). $$
\end{proposition}
\begin{proof}
We only need to prove the first equality, for which we prove both inequalities:

$ \leq $ : Let $ \mathcal{A} $ be a P-basis of $ \Omega $ and define $ \mathcal{A}_F = \{ a \in \mathcal{A} \mid F(a) = 0 \} $. Since $ \mathcal{A}_F \subseteq \mathcal{A} $, then $ \mathcal{A}_F $ is P-independent. Together with $ \mathcal{A}_F \subseteq Z_\Omega(F) $, we deduce that $ \# \mathcal{A}_F \leq {\rm Rk}(Z_\Omega(F)) $. Using that $ {\rm Rk}(\Omega) = \# \mathcal{A} $, we conclude that
$$ {\rm wt}_\Omega (F) = {\rm Rk}(\Omega) - {\rm Rk}(Z_\Omega(F)) \leq \# \mathcal{A} - \# \mathcal{A}_F = {\rm wt_H}(E_\mathcal{A}(F)), $$
and the inequality follows.

$ \geq $ : Let $ \mathcal{A}_F $ be a P-basis of $ Z_\Omega(F) $. By Corollary \ref{cor charact P-bases}, there exists a P-basis $ \mathcal{A} $ of $ \Omega $ that contains $ \mathcal{A}_F $. Therefore
$$ {\rm wt_H}(E_\mathcal{A}(F)) = \# \mathcal{A} - \# \mathcal{A}_F = {\rm Rk}(\Omega) - {\rm Rk}(Z_\Omega(F)) = {\rm wt}_\Omega (F), $$
and the inequality follows.
\end{proof}

\subsection{Maximum skew distance codes and skew Reed-Solomon codes} \label{subsec optimal codes}

In this subsection, we deduce a Singleton bound for skew metrics and show how the skew Reed-Solomon codes defined in \cite{skew-evaluation1} for finite fields always attain such bounds over arbitrary division rings. Moreover by Proposition \ref{prop connection with Hamming weights}, we will see that any maximum distance left linear code for skew metrics is maximum distance separable (MDS), that is, it is maximum distance for the Hamming metric. Thus our results imply Item 1 in \cite[Proposition 2]{skew-evaluation1}.

Given an arbitrary (linear or non-linear) code $ \mathcal{C} \subseteq \mathbb{F}^\mathcal{B} $, we define its minimum skew distance by
\begin{equation}
{\rm d}_\mathcal{B}^{\sigma, \delta}(\mathcal{C}) = \min \{ {\rm d}_\mathcal{B}^{\sigma,\delta}(f, g) \mid f,g \in \mathcal{C}, f \neq g \}.
\end{equation}
We may give analogous definitions in $ \mathbb{F}[x; \sigma, \delta]_n $ with the metric $ {\rm d}_\Omega^{\sigma, \delta} $. Again, we use the notation $ {\rm d}_\mathcal{B} $ and $ {\rm d}_\Omega $ whenever $ \sigma $ and $ \delta $ are known from the context.

As usual, it holds that $ {\rm d}_\mathcal{B}(\mathcal{C}) = \min \{ {\rm wt}_\mathcal{B}(f) \mid f \in \mathcal{C}, f \neq 0 \} $ if $ \mathcal{C} \subseteq \mathbb{F}^\mathcal{B} $ is left linear. The following lemma follows directly from Proposition \ref{prop connection with Hamming weights}:

\begin{lemma} \label{lemma connection with Hamming metric}
Given an arbitrary code $ \mathcal{C} \subseteq \mathbb{F}^\mathcal{B} $, it holds that
$$ {\rm d}_\mathcal{B}(\mathcal{C}) = \min \{ {\rm d}_H(\pi_{\mathcal{B},\mathcal{A}}(\mathcal{C})) \mid \mathcal{A} \textrm{ is a P-basis of } \Omega \} \leq {\rm d}_H(\mathcal{C}). $$
\end{lemma}

Thus the Singleton bound for the Hamming metric implies the Singleton bound for skew metrics:

\begin{proposition} \label{prop singleton bound}
Let $ \mathcal{C} \subseteq \mathbb{F}^\mathcal{B} $ be a left linear code of dimension $ k $. It holds that
\begin{equation}
{\rm d}_\mathcal{B}(\mathcal{C}) \leq n - k + 1.
\label{eq singleton bound}
\end{equation}
\end{proposition}

This motivates the following definition:

\begin{definition} [\textbf{Maximum skew distance codes}]
We say that a left linear code $ \mathcal{C} \subseteq \mathbb{F}^\mathcal{B} $ is a maximum skew distance (MSD) code if equality holds in (\ref{eq singleton bound}).
\end{definition}

Thus we obtain the following direct consequence from Lemma \ref{lemma connection with Hamming metric}:

\begin{corollary}
A left linear code $ \mathcal{C} \subseteq \mathbb{F}^\mathcal{B} $ is MSD if, and only if, $ \pi_{\mathcal{B},\mathcal{A}}(\mathcal{C}) \subseteq \mathbb{F}^\mathcal{A} $ is MDS, for all P-bases $ \mathcal{A} $ of $ \Omega $. 
\end{corollary}

We will now describe an explicit family of maximum skew distance codes for any dimension $ k = 1,2, \ldots, n $. These codes have been first defined in \cite[Definition 7]{skew-evaluation1} when $ \mathbb{F} $ is finite. Extending the definition is straightforward.

\begin{definition} [\textbf{Skew Reed-Solomon codes \cite{skew-evaluation1}}] \label{def skew RS codes}
For each $ k = 0,1,2, \ldots, n $, we define the ($ k $-dimensional) skew Reed-Solomon code over the P-basis $ \mathcal{B} $ with endomorphism $ \sigma $ and derivation $ \delta $ as the linear code
$$ \mathcal{C}_{\mathcal{B},k}^{\sigma, \delta} = E_\mathcal{B}^{\sigma,\delta}(\mathbb{F}[x; \sigma, \delta]_k) \subseteq \mathbb{F}^\mathcal{B} . $$
Again, we will write $ \mathcal{C}_{\mathcal{B},k} $ whenever $ \sigma $ and $ \delta $ are known from the context.
\end{definition}

The following result implies Item 1 in \cite[Proposition 2]{skew-evaluation1}, which proves that skew Reed-Solomon codes over finite fields are MDS.

\begin{theorem} \label{th maximum skew distance codes}
For each $ k = 1,2, \ldots, n $, it holds that
$$ {\rm d}_\mathcal{B}(\mathcal{C}_{\mathcal{B},k}) = n - k + 1. $$
That is, the code $ \mathcal{C}_{\mathcal{B},k} $ is a maximum skew distance code. In particular, it is also MDS.
\end{theorem}
\begin{proof}
Let $ F \in \mathbb{F}[x; \sigma, \delta]_k $ be such that $ {\rm Rk}(\Omega) - {\rm Rk}(Z_\Omega(F)) \leq n-k $, that is,
$$ {\rm Rk}(Z_\Omega(F)) \geq k. $$
By Lemma \ref{lemma lagrange interpolation}, it holds that $ F = 0 $. We deduce that $ {\rm d}_\mathcal{B}(\mathcal{C}_{\mathcal{B},k}) \geq n - k + 1 $, and since the reversed inequality always holds by Proposition \ref{prop singleton bound}, the result follows.
\end{proof}

We conclude by observing that we obtain again skew Reed-Solomon codes by changing P-basis: If $ 0 \leq k \leq n $ and $ \mathcal{A} $ is another P-basis of $ \Omega $, then
$$ \mathcal{C}_{\mathcal{A},k} = \pi_{\mathcal{B}, \mathcal{A}}(\mathcal{C}_{\mathcal{B},k}) \subseteq \mathbb{F}^\mathcal{A}. $$

\section{Linearizing skew metrics and skew Reed-Solomon codes} \label{sec linearizing skew metrics}

In this section, we recall the concept of \textit{linear operator polynomials} from \cite{leroy-pol} and use it to give linearized descriptions of skew metrics and skew Reed-Solomon codes, thus obtaining analogous descriptions as that of Gabidulin codes \cite{gabidulin, new-construction}. The linear operators $ \sigma $ or $ \delta $ have been considered in \cite{augot-function, augot, augot-extended, skew-evaluation1, gabidulin, new-construction, lam, lam-leroy, skew-evaluation2, matroidal}. General linear operator polynomials have been introduced in \cite{leroy-pol}. Those that we are interested in were defined in \cite[Example 2.6]{leroy-pol}, although the right linear functions that they define were already considered in \cite{hilbert90}. Their connection with skew Reed-Solomon codes \cite{skew-evaluation1} in general is new to the best of our knowledge. Furthermore, the general definition of \textit{linearized Reed-Solomon codes} is new.

\subsection{Linear operators and operator polynomials} \label{subsec linear operators and RS}

In this subsection, we recall the definitions of (linear) operator polynomials and their connection with skew polynomials. The following definition is \cite[Example 2.6]{leroy-pol}:

\begin{definition} [\textbf{Operator polynomials \cite{leroy-pol}}]
Given $ a \in \mathbb{F} $, we define its $ (\sigma,\delta) $-operator as
\begin{equation*}
\begin{array}{rccc}
\mathcal{D}_a^{\sigma, \delta} : & \mathbb{F} & \longrightarrow & \mathbb{F} \\
 & b & \mapsto & \sigma(b)a + \delta(b),
\end{array}
\end{equation*}
for all $ b \in \mathbb{F} $. We will denote $ \mathcal{D}_a $ or $ \mathcal{D} $ for simplicity when $ \sigma, \delta $ or $ \sigma, \delta, a $ are known from the context, respectively. We then define the left vector space of operator polynomials $ \mathbb{F}[\mathcal{D}_a] $ over $ \mathbb{F} $ as that with basis $ \{ \mathcal{D}_a^i \mid i \in \mathbb{N} \} $, where $ \mathcal{D}_a^i $ are just pair-wise distinct symbols. For an operator polynomial $ F = \sum_{i \in \mathbb{N}} F_i \mathcal{D}_a^i $ with $ F_i \in \mathbb{F} $ for all $ i \in \mathbb{N} $, we define its operator evaluation over $ b \in \mathbb{F} $ as
$$ F(b) = \sum_{i \in \mathbb{N}} F_i \mathcal{D}_a^i(b), $$
for all $ b \in \mathbb{F} $, where $ \mathcal{D}_a^0 = {\rm Id} $. Given a set $ \Omega \subseteq \mathbb{F} $, we define the evaluation map over $ \Omega $ as the left linear map
$$ E_\Omega : \mathbb{F}[\mathcal{D}_a] \longrightarrow \mathbb{F}^\Omega $$
where $ f = E_\Omega(F) \in \mathbb{F}^\Omega $ is given by $ f(b) = F(b) $, for all $ b \in \Omega $ and for $ F \in \mathbb{F}[\mathcal{D}_a] $.

Finally, we denote by $ F^{\mathcal{D}_a} $ the image of $ F \in \mathbb{F}[x; \sigma,\delta] $ by the left vector space isomorphism
\begin{equation*}
\begin{array}{ccc}
 \mathbb{F}[x; \sigma,\delta] & \longrightarrow & \mathbb{F}[\mathcal{D}_a^{\sigma,\delta}] \\
 \sum_{i \in \mathbb{N}} F_i x^i & \mapsto & \sum_{i \in \mathbb{N}} F_i \mathcal{D}_a^i.
\end{array}
\end{equation*}
\end{definition}

Operator polynomials are right linear over certain division subring of $ \mathbb{F} $, which is defined as the \textit{centralizer} of $ a $ in \cite[Eq. (3.1)]{lam-leroy}:

\begin{definition} [\textbf{Centralizers \cite{lam-leroy}}]
Given $ a \in \mathbb{F} $, we define its $ (\sigma,\delta) $-centralizer, or simply centralizer, as
$$ K_a = K_a^{\sigma,\delta} = \{ b \in \mathbb{F} \mid \mathcal{D}^{\sigma,\delta}_a(b) = ab \}. $$
\end{definition} 

The following lemma is a particular case of \cite[Lemma 3.2]{lam-leroy}:

\begin{lemma} [\textbf{\cite{lam-leroy}}]
For all $ a \in \mathbb{F} $, it holds that $ K_a \subseteq \mathbb{F} $ is a division subring of $ \mathbb{F} $.
\end{lemma}

The proof of the next result is straightforward and is also given in \cite[Section 3]{hilbert90}:

\begin{lemma} [\textbf{\cite{hilbert90}}] \label{lemma operator polynomials right linear}
Given $ a \in \mathbb{F} $ and $ F \in \mathbb{F}[\mathcal{D}_a] $, the map $ b \mapsto F(b) $, for $ b \in \mathbb{F} $, is right linear over $ K_a $. That is, for all $ b,c \in \mathbb{F} $ and all $ \lambda \in K_a $, it holds that
$$ F(b+c) = F(b) + F(c) \quad \textrm{and} \quad F(b \lambda) = F(b) \lambda. $$
\end{lemma}

We now give the main connection between skew polynomial evaluation and linear operator polynomial evaluation. This result can be found in a more general form at the end of \cite[Theorem 2.8]{leroy-pol}. We give in Appendix \ref{app proof of th evaluation} an alternative proof using the language in our paper.

\begin{lemma} [\textbf{\cite{leroy-pol}}] \label{lemma evaluation of skew as operator}
Given $ a \in \mathbb{F} $, $ b \in \mathbb{F}^* $ and $ F \in \mathbb{F}[x;\sigma,\delta] $, and writing $ \mathcal{D} = \mathcal{D}_a $, it holds that
$$ F(\mathcal{D}(b)b^{-1}) = F^\mathcal{D}(b)b^{-1}. $$
\end{lemma}

The functions defined by evaluating linear operator polynomials were already considered in \cite[Definition 3.1]{hilbert90}, without directly using the previous lemma.

\subsection{Sum-rank metrics from linearizing skew metrics} \label{subsec sum rank metric}

In this subsection, we will establish the connection between the skew metrics defined in Subsection \ref{subsec skew pol metrics} and sum-rank metrics \cite{multishot}. Our definition of sum-rank metrics is more general than that given in \cite{multishot}. The definition in \cite{multishot} is recovered by choosing a finite field $ \mathbb{F} $, and assuming $ n_1 = n_2 = \ldots = n_\ell $ and $ K_1 = K_2 = \ldots = K_\ell $.

For a division subring $ K \subseteq \mathbb{F} $ and a subset $ \mathcal{A} \subseteq \mathbb{F} $, we will denote by $ \langle \mathcal{A} \rangle_K^R $ the right subspace over $ K $ generated by $ \mathcal{A} $. We use $ \dim_K $ to denote its dimension over $ K $.

\begin{definition} [\textbf{Sum-rank metrics}] \label{def sum rank metric}
Let $ K_1, K_2, \ldots, K_\ell $ be division subrings of $ \mathbb{F} $. Given positive integers $ n_1, n_2, \ldots, n_\ell $ and $ n = n_1 + n_2 + \cdots + n_\ell $, we define the sum-rank weight in $ \mathbb{F}^n $ with lengths $ (n_1, n_2, \ldots, n_\ell) $ and division subrings $ (K_1, K_2, \ldots, K_\ell) $ as
$$ {\rm wt_{SR}}(\mathbf{c}) = {\rm wt_R}(\mathbf{c}^{(1)}) + {\rm wt_R}(\mathbf{c}^{(2)}) + \cdots + {\rm wt_R}(\mathbf{c}^{(\ell)}), $$
where $ \mathbf{c} = (\mathbf{c}^{(1)}, \mathbf{c}^{(2)}, \ldots, \mathbf{c}^{(\ell)}) \in \mathbb{F}^n $ and $ \mathbf{c}^{(i)} \in \mathbb{F}^{n_i} $, and where
$$ {\rm wt_R}(\mathbf{c}^{(i)}) = \dim_{K_i}( \langle c_1^{(i)}, c_2^{(i)}, \ldots, c_{n_i}^{(i)} \rangle_{K_i}^R), $$
for $ i = 1,2, \ldots, \ell $. We define the associated metric $ {\rm d_{SR}} : (\mathbb{F}^n)^2 \longrightarrow \mathbb{N} $ by 
$$ {\rm d_{SR}}(\mathbf{c}, \mathbf{d}) = {\rm wt_{SR}}(\mathbf{c} - \mathbf{d}), $$
for all $ \mathbf{c}, \mathbf{d} \in \mathbb{F}^n $.
\end{definition}

Observe that $ {\rm wt_{SR}} $ and $ {\rm d_{SR}} $ are indeed a weight and a metric, respectively, in $ \mathbb{F}^n $. Observe that sum-rank metrics extend the Hamming metric by choosing $ n_1 = n_2 = \ldots = n_\ell = 1 $, and they extend the rank metric by choosing $ \ell = 1 $.

To establish their relation with skew metrics, we need some results from \cite{lam} and \cite[Section 4]{lam-leroy} concerning how to (right) linearize the concept of P-independence. The following lemma is \cite[Theorem 4.5]{lam-leroy}:

\begin{lemma} [\textbf{\cite{lam-leroy}}] \label{lemma linearized in one conj class}
Fix $ a \in \mathbb{F} $ and $ \alpha_1, \alpha_2, \ldots, \alpha_n \in \mathbb{F}^* $ and define
$$ a_i = \mathcal{D}_a(\alpha_i)\alpha_i^{-1}, $$
for $ i = 1,2, \ldots, n $. If $ \Phi = \overline{\{ a_1, a_2, \ldots, a_n \}} $, then
$$ {\rm Rk}(\Phi) = \dim_{K_a}( \langle \alpha_1, \alpha_2, \ldots, \alpha_n \rangle_{K_a}^R ). $$
\end{lemma}

Hence we deduce the following:

\begin{corollary} \label{cor linearized in one conj class}
Fix $ a \in \mathbb{F} $, let $ \alpha_i, \beta_i \in \mathbb{F}^* $ and define
$$ a_i = \mathcal{D}_a(\alpha_i)\alpha_i^{-1} \quad \textrm{and} \quad b_i = \mathcal{D}_a(\beta_i)\beta_i^{-1}, $$
for $ i = 1,2, \ldots, n $. Then the following are equivalent:
\begin{enumerate}
\item
$ \mathcal{A} = \{ a_1, a_2, \ldots, a_n \} $ and $ \mathcal{B} = \{ b_1, b_2, \ldots, b_n \} $ are P-bases of the same P-closed set $ \Omega $.
\item
$ \mathcal{A}_\mathcal{D} = \{ \alpha_1, \alpha_2, \ldots, \alpha_n \} $ and $ \mathcal{B}_\mathcal{D} = \{ \beta_1, \beta_2, \ldots, \beta_n \} $ are bases of the same right vector space over $ K_a $.
\item
$ \alpha_1, \alpha_2, \ldots, \alpha_n $ are right linearly independent over $ K_a $ and there exists an invertible matrix $ A \in K_a^{n \times n} $ (there exists $ B \in K_a^{n \times n} $ such that $ AB = BA = I $) such that
\begin{equation}
\boldsymbol\beta = \boldsymbol\alpha A,
\label{eq linearized change of P-basis}
\end{equation}
where $ \boldsymbol\alpha = (\alpha_1, \alpha_2, \ldots, \alpha_n) \in \mathbb{F}^n $ and $ \boldsymbol\beta = (\beta_1, \beta_2, \ldots, \beta_n) \in \mathbb{F}^n $.
\end{enumerate}
\end{corollary}
\begin{proof}
Items 2 and 3 are equivalent by basic right linear algebra. 

Assume now that Item 1 holds, therefore the sets $ \mathcal{A}_\mathcal{D} $ and $ \mathcal{B}_\mathcal{D} $ are right linearly independent by the previous lemma. If Item 2 does not hold, then there exists $ 1 \leq j \leq n $ such that $ \beta_j \notin \langle \alpha_1, \alpha_2, \ldots, \alpha_n \rangle_{K_a}^R $. Hence $ \dim_{K_a}(\langle \beta_j, \alpha_1, \alpha_2, \ldots, \alpha_n \rangle_{K_a}^R) = n+1 $. By the previous lemma we deduce that $ {\rm Rk}(\Omega) \geq n+1 $, which is a contradiction since $ \mathcal{A} $ is a P-basis of $ \Omega $ and has $ n $ elements.

Assume now that Item 2 holds, therefore the sets $ \mathcal{A} $ and $ \mathcal{B} $ are P-independent by the previous lemma. If Item 1 does not hold, then there exists $ 1 \leq j \leq n $ such that $ b_j \notin \overline{\mathcal{A}} $. Hence $ \mathcal{A}^\prime = \mathcal{A} \cup \{ b_j \} $ is P-independent, thus $ {\rm Rk}(\Phi) = n+1 $, where $ \Phi = \overline{\mathcal{A}^\prime} $. By the previous lemma we deduce that $ \dim_{K_a}(\langle \beta_j, \alpha_1, \alpha_2, \ldots, \alpha_n \rangle_{K_a}^R) = n+1 $, which contradicts Item 2.
\end{proof}

To be able to use different elements $ a \in \mathbb{F} $, we need the concept of conjugacy \cite{lam, lam-leroy}:

\begin{definition} [\textbf{Conjugacy \cite{lam, lam-leroy}}]
We define the conjugacy relation $ \sim $ in $ \mathbb{F} $ by $ a \sim c $ if $ c = \mathcal{D}_a(b)b^{-1} = \sigma(b)ab^{-1} + \delta(b)b^{-1} $, for some $ b \in \mathbb{F}^* $. Given $ a \in \mathbb{F} $, we define its conjugacy class as
$$ C(a) = \{ \mathcal{D}_a(b)b^{-1} \mid b \in \mathbb{F}^* \}. $$
\end{definition}

It is easy to check that $ \sim $ is an equivalence relation in $ \mathbb{F} $, thus conjugacy classes give a partition of $ \mathbb{F} $. The importance of conjugacy is given by the following result, which is \cite[Theorem 23]{lam} (see also \cite[Section 4]{lam-leroy}):

\begin{lemma} [\textbf{\cite{lam}}] \label{lemma ranks partition conjugacy}
Let $ a^{(1)}, a^{(2)}, \ldots, a^{(\ell)} \in \mathbb{F} $ belong to pair-wise distinct conjugacy classes. Let $ \mathcal{B}_i \subseteq C(a^{(i)}) $ be a finite set and define $ \Omega_i = \overline{\mathcal{B}_i} $, for $ i = 1,2, \ldots, \ell $. If $ \mathcal{B} = \mathcal{B}_1 \cup \mathcal{B}_2 \cup \ldots \cup \mathcal{B}_\ell $ and $ \Omega = \overline{\mathcal{B}} $, then it holds that 
$$ {\rm Rk}(\Omega) = {\rm Rk}(\Omega_1) + {\rm Rk}(\Omega_2) + \cdots + {\rm Rk}(\Omega_\ell). $$
In particular, $ \mathcal{B} $ is P-independent if, and only if, so is $ \mathcal{B}_i $ for all $ i = 1,2, \ldots, \ell $.
\end{lemma}

Hence we may deduce the following fact on the partition in conjugacy classes of P-bases of a P-closed set:

\begin{corollary}
Let $ \Omega \subseteq \mathbb{F} $ be a P-closed set with P-basis $ \mathcal{B} $. If $ a \in \mathbb{F} $ and $ \mathcal{B}_a = \mathcal{B} \cap C(a) $, then it holds that
$$ \overline{\mathcal{B}_a} = \Omega \cap C(a). $$
\end{corollary}
\begin{proof}
First, $ \Omega \cap C(a) $ is P-closed by the following argument: $ F_{\Omega \cap C(a)} $ divides $ F_\Omega $ on the right, hence is not zero. By that fact and the product rule \cite[Theorem 2.7]{lam-leroy}, every root of $ F_{\Omega \cap C(a)} $ lies in $ C(a) $, hence $ Z(F_{\Omega \cap C(a)}) \subseteq \Omega \cap C(a) $.

Next, since they belong to different conjugacy classes, we have by Lemma \ref{lemma ranks partition conjugacy} that
$$ \sum_{a \in \mathbb{F}} {\rm Rk}(\Omega \cap C(a)) \leq {\rm Rk}(\Omega), $$
running over disjoint conjugacy classes. By Lemma \ref{lemma ranks partition conjugacy}, we also have that
$$ {\rm Rk}(\Omega) = \sum_{a \in \mathbb{F}} {\rm Rk}(\overline{\mathcal{B}_a}), $$
again running over disjoint conjugacy classes. Since $ \overline{\mathcal{B}_a} \subseteq \Omega \cap C(a) $, we conclude that $ {\rm Rk}(\overline{\mathcal{B}_a}) = {\rm Rk}(\Omega \cap C(a)) $, and the result follows.
\end{proof}

We may now establish the first main result of this section. This result gives a correspondence between the isometries for the skew metric in Definition \ref{def change of P-basis} and some isometries for the sum-rank metric:

\begin{theorem} \label{th big commutative diagram}
Let $ \Omega \subseteq \mathbb{F} $ be a P-closed set with P-bases $ \mathcal{A} $ and $ \mathcal{B} $, and let 
$$ \mathcal{A}_i = \mathcal{A} \cap C(a^{(i)}) \quad \textrm{and} \quad \mathcal{B}_i = \mathcal{B} \cap C(a^{(i)}), $$
for $ i = 1,2, \ldots, \ell $, be non-empty pair-wise disjoint P-independent sets with $ \mathcal{A} = \mathcal{A}_1 \cup \mathcal{A}_2 \cup \ldots \cup \mathcal{A}_\ell $ and $ \mathcal{B} = \mathcal{B}_1 \cup \mathcal{B}_2 \cup \ldots \cup \mathcal{B}_\ell $. For each $ i = 1,2, \ldots, \ell $, let $ \mathcal{A}_i = \{ a_1^{(i)}, a_2^{(i)}, \ldots, a_{n_i}^{(i)} \} $ and $ \mathcal{B}_i = \{ b_1^{(i)}, b_2^{(i)}, \ldots, b_{n_i}^{(i)} \} $ (both have the same size by the previous corollary), and let $ \alpha_j^{(i)}, \beta_j^{(i)} \in \mathbb{F}^* $ be such that
$$ a_j^{(i)} = \mathcal{D}_{a^{(i)}} \left( \alpha_j^{(i)} \right) \left( \alpha_j^{(i)} \right)^{-1} \quad \textrm{and} \quad b_j^{(i)} = \mathcal{D}_{a^{(i)}} \left( \beta_j^{(i)} \right) \left( \beta_j^{(i)} \right)^{-1}, $$
for $ j = 1,2, \ldots, n_i $. Finally, write $ n = {\rm Rk}(\Omega) = n_1 + n_2 + \cdots + n_\ell $ and define the left linear maps $ \phi_\mathcal{A} : \mathbb{F}^\mathcal{A} \longrightarrow \mathbb{F}^n $ and $ \phi_\mathcal{B} : \mathbb{F}^\mathcal{B} \longrightarrow \mathbb{F}^n $ by $ \phi_\mathcal{A}(f) = (\mathbf{c}^{(1)}, \mathbf{c}^{(2)}, \ldots, \mathbf{c}^{(\ell)}) $, where $ \mathbf{c}^{(i)} \in \mathbb{F}^{n_i} $ and
$$ c^{(i)}_j = F^{\mathcal{D}_{a^{(i)}}} \left( \alpha_j^{(i)} \right) , $$
for all $ f \in \mathbb{F}^\mathcal{A} $, where $ f = E_\mathcal{A}(F) $ and $ F \in \mathbb{F}[x; \sigma,\delta]_n $, for $ j = 1,2, \ldots, n_i $ and $ i = 1,2, \ldots, \ell $. Analogously for $ \mathcal{B} $. Then the following diagram is commutative, where all maps are left vector space isomorphisms:
\begin{displaymath}
\begin{array}{rcccl}
 & \mathbb{F}^\mathcal{A} & \stackrel{\pi_{\mathcal{A},\mathcal{B}}}{\longrightarrow} & \mathbb{F}^\mathcal{B} & \\
\phi_\mathcal{A} & \downarrow &  & \downarrow & \phi_\mathcal{B} \\
 & \mathbb{F}^n & \stackrel{\pi_{A}}{\longrightarrow} & \mathbb{F}^n & \\
\end{array}
\end{displaymath}
where $ \pi_A(\mathbf{c}) = \mathbf{c} A $, for $ \mathbf{c} \in \mathbb{F}^n $, and 
\begin{displaymath}
A = \left( \begin{array}{cccc}
A_1 & 0 & \ldots & 0 \\
0 & A_2 & \ldots & 0 \\
\vdots & \vdots & \ddots & \vdots \\
0 & 0 & \ldots & A_\ell \\
\end{array} \right) \in \mathbb{F}^{n \times n},
\end{displaymath}
where $ A_i \in K_{a^{(i)}}^{n_i \times n_i} $ is invertible and is given as in (\ref{eq linearized change of P-basis}) for $ \boldsymbol\alpha_i = (\alpha_1^{(i)}, \alpha_2^{(i)}, \ldots, \alpha_{n_i}^{(i)}) \in \mathbb{F}^{n_i} $ and $ \boldsymbol\beta_i = (\beta_1^{(i)}, \beta_2^{(i)}, \ldots, \beta_{n_i}^{(i)}) \in \mathbb{F}^{n_i} $, for $ i = 1,2, \ldots, \ell $.
\end{theorem}
\begin{proof}
It is trivial to see that $ \phi_\mathcal{A} $, $ \phi_\mathcal{B} $ and $ \pi_A $ are left linear. Moreover, $ \pi_A $ is invertible by using inverses of each $ A_i \in K_{a^{(i)}}^{n_i \times n_i} $, for $ i = 1,2, \ldots, \ell $. We now show that $ \phi_\mathcal{A} $ and $ \phi_\mathcal{B} $ are left vector space isomorphisms, for which we only need to show that they are one to one. 

Take $ f \in \mathbb{F}^\mathcal{A} $ and let $ F \in \mathbb{F}[x;\sigma,\delta]_n $ be such that $ f = E_\mathcal{A}(F) $. If $ \phi_\mathcal{A}(f) = \mathbf{0} $, then by definition $ F^{\mathcal{D}_{a^{(i)}}}(\alpha_j^{(i)}) = 0 $, hence by Lemma \ref{lemma evaluation of skew as operator} it holds that $ f(a_j^{(i)}) = F(a_j^{(i)}) = 0 $, for all $ j = 1,2, \ldots, n_i $ and all $ i = 1,2, \ldots, \ell $. Thus $ f = 0 $ and we are done. Similarly for $ \phi_\mathcal{B} $.

We will now show that the given diagram is commutative. Let again $ f \in \mathbb{F}^\mathcal{A} $ and $ F \in \mathbb{F}[x;\sigma,\delta]_n $ be such that $ f = E_\mathcal{A}(F) $. Define $ g = \pi_{\mathcal{A},\mathcal{B}}(f) $, and
$$ \phi_\mathcal{A}(f) = (\mathbf{c}^{(1)}, \mathbf{c}^{(2)}, \ldots, \mathbf{c}^{(\ell)}), $$
$$ \phi_\mathcal{B}(g) = (\mathbf{d}^{(1)}, \mathbf{d}^{(2)}, \ldots, \mathbf{d}^{(\ell)}), $$
where $ \mathbf{c}^{(i)}, \mathbf{d}^{(i)} \in \mathbb{F}^{n_i} $, for $ i = 1,2, \ldots, \ell $. By definition, it holds that
$$ c^{(i)}_j = F^{\mathcal{D}_{a^{(i)}}} \left( \alpha_j^{(i)} \right) \quad \textrm{and} \quad d^{(i)}_k = F^{\mathcal{D}_{a^{(i)}}} \left( \beta_k^{(i)} \right), $$
for $ j,k = 1,2, \ldots, n_i $ and for $ i = 1,2, \ldots, \ell $. If we denote by $ \lambda_{j,k} \in K_{a^{(i)}} $ the $ (j,k) $-th entry of the matrix $ A_i $ (we omit the superindex $ i $ for brevity), then we have that
$$ \beta_k^{(i)} = \sum_{j=1}^{n_i} \alpha_j^{(i)} \lambda_{j,k}, $$
and thus by Lemma \ref{lemma operator polynomials right linear}, we have that
$$ d^{(i)}_k = F^{\mathcal{D}_{a^{(i)}}} \left( \sum_{j=1}^{n_i} \alpha_j^{(i)} \lambda_{j,k} \right) = \sum_{j=1}^{n_i} F^{\mathcal{D}_{a^{(i)}}} \left( \alpha_j^{(i)} \right) \lambda_{j,k} = \sum_{j=1}^{n_i} c^{(i)}_j \lambda_{j,k}, $$
for all $ k = 1,2, \ldots, n_i $ and all $ i = 1,2, \ldots, \ell $. Hence
$$ \phi_\mathcal{A}(f) A = \phi_\mathcal{B}(\pi_{\mathcal{A}, \mathcal{B}}(f)), $$
and the given diagram is commutative.
\end{proof}

Thus we obtain the second main result of this section, which is the above mentioned connection between skew metrics and sum-rank metrics:

\begin{theorem} \label{th both weights coincide}
Let the notation be as in Theorem \ref{th big commutative diagram} and denote $ K_i = K_{a^{(i)}} $, for $ i = 1,2, \ldots, \ell $. Then the weight $ {\rm wt} : \mathbb{F}^n \longrightarrow \mathbb{N} $ given by
$$ {\rm wt}(\mathbf{c}) = {\rm wt}_\mathcal{A}(f), $$
for all $ f \in \mathbb{F}^\mathcal{A} $, where $ \mathbf{c} = \phi_\mathcal{A}(f) $, is the sum-rank weight in $ \mathbb{F}^n $ with lengths $ (n_1, n_2, \ldots, n_\ell) $ and division subrings $ (K_1, K_2, \ldots, K_\ell) $ (see Definition \ref{def sum rank metric}).
\end{theorem}
\begin{proof}
By combining Proposition \ref{prop connection with Hamming weights} and Theorem \ref{th big commutative diagram}, we deduce that
$$ {\rm wt}(\mathbf{c}) = \min \left\lbrace \sum_{i=1}^\ell {\rm wt_H}(\mathbf{c}^{(i)} A_i) \mid A_i \in K_i^{n_i \times n_i} \textrm{ invertible, for } i = 1,2, \ldots, \ell \right\rbrace . $$
By a linear algebra argument, we also have that
$$ {\rm wt_R}(\mathbf{d}) = \dim_{K_i}( \langle d_1, d_2, \ldots, d_{n_i} \rangle_{K_i}^R) = \min \{ {\rm wt_H}(\mathbf{d} B) \mid B \in K_i^{n_i \times n_i} \textrm{ invertible} \}, $$
for $ \mathbf{d} = (d_1, d_2, \ldots, d_{n_i}) \in \mathbb{F}^{n_i} $ and for $ i = 1,2, \ldots, \ell $. The result follows by combining these two facts.
\end{proof}

\subsection{Linearized Reed-Solomon codes with maximum sum-rank distance} \label{subsec linearized RS codes}

In this subsection, we define general linearized Reed-Solomon codes, establish their connection with skew Reed-Solomon codes, and deduce from the previous study that they have maximum sum-rank distance. We conclude by recalling what particular cases of linearized Reed-Solomon codes have been given in the literature.

\begin{definition} [\textbf{Linearized Reed-Solomon codes}] \label{def linearized RS codes}
Let the notation be as in Theorem \ref{th big commutative diagram} and denote $ K_i = K_{a^{(i)}} $, for $ i = 1,2, \ldots, \ell $. For each $ k = 0,1,2, \ldots, n $, we define the ($ k $-dimensional) linearized Reed-Solomon code with conjugacy representatives $ a^{(1)}, a^{(2)}, \ldots, a^{(\ell)} \in \mathbb{F} $ and basis vectors $ \boldsymbol\beta_i = (\beta_1^{(i)}, \beta_2^{(i)}, \ldots, \beta_{n_i}^{(i)}) \in \mathbb{F}^{n_i} $, for $ i = 1,2, \ldots, \ell $, as the left linear code $ \mathcal{C}^{\sigma,\delta}_{L,k} \subseteq \mathbb{F}^n $ formed by the vectors $ \mathbf{c} = (\mathbf{c}^{(1)}, \mathbf{c}^{(2)}, \ldots, \mathbf{c}^{(\ell)}) \in \mathbb{F}^n $ given by $ \mathbf{c}^{(i)} = (c_1^{(i)}, c_2^{(i)}, \ldots, c_{n_i}^{(i)}) \subseteq \mathbb{F}^{n_i} $ and
$$ c_j^{(i)} = \sum_{l=0}^{k-1} F_l \mathcal{D}_{a^{(i)}}^l \left( \beta_j^{(i)} \right), $$
where $ F_l \in \mathbb{F} $, for $ l = 0,1,2, \ldots, k-1 $, for $ j = 1,2, \ldots, n_i $, and for $ i = 1,2, \ldots, \ell $. We use the notation $ \mathcal{C}_{L,k} $ when $ \sigma $ and $ \delta $ are understood from the context.
\end{definition}

Observe that these codes depend on the conjugacy representatives and the basis vectors. However, we omit this in the notation for brevity.

Observe also that a generator matrix for $ \mathcal{C}^{\sigma,\delta}_{L,k} \subseteq \mathbb{F}^n $ is then given by the matrix formed by evaluations on the different operators $ \mathcal{D}_i = \mathcal{D}_{a^{(i)}} $, for $ i = 1,2, \ldots, \ell $:
\begin{equation*}
\scalebox{0.85}{$
\left( \begin{array}{cccc|c|cccc}
\beta_1^{(1)} & \beta_2^{(1)} & \ldots & \beta_{n_1}^{(1)} & \ldots & \beta_1^{(\ell)} & \beta_2^{(\ell)} & \ldots & \beta_{n_\ell}^{(\ell)} \\
\mathcal{D}_1 \left( \beta_1^{(1)} \right) & \mathcal{D}_1 \left( \beta_2^{(1)} \right) & \ldots & \mathcal{D}_1 \left( \beta_{n_1}^{(1)} \right) & \ldots & \mathcal{D}_\ell \left( \beta_1^{(\ell)} \right) & \mathcal{D}_\ell \left( \beta_2^{(\ell)} \right) & \ldots & \mathcal{D}_\ell \left( \beta_{n_\ell}^{(\ell)} \right) \\
\mathcal{D}_1^2 \left( \beta_1^{(1)} \right) & \mathcal{D}_1^2 \left( \beta_2^{(1)} \right) & \ldots & \mathcal{D}_1^2 \left( \beta_{n_1}^{(1)} \right) & \ldots & \mathcal{D}_\ell^2 \left( \beta_1^{(\ell)} \right) & \mathcal{D}_\ell^2 \left( \beta_2^{(\ell)} \right) & \ldots & \mathcal{D}_\ell^2 \left( \beta_{n_\ell}^{(\ell)} \right) \\
\vdots & \vdots & \ddots & \vdots & \ddots & \vdots & \vdots & \ddots & \vdots \\
\mathcal{D}_1^{k-1} \left( \beta_1^{(1)} \right) & \mathcal{D}_1^{k-1} \left( \beta_2^{(1)} \right) & \ldots & \mathcal{D}_1^{k-1} \left( \beta_{n_1}^{(1)} \right) & \ldots & \mathcal{D}_\ell^{k-1} \left( \beta_1^{(\ell)} \right) & \mathcal{D}_\ell^{k-1} \left( \beta_2^{(\ell)} \right) & \ldots & \mathcal{D}_\ell^{k-1} \left( \beta_{n_\ell}^{(\ell)} \right) \\
\end{array} \right).
$}
\end{equation*}

We may compute such powers of the operators $ \mathcal{D}_a $ by the following iterative formula, which is trivial from the definition:

\begin{proposition} \label{prop iterative formula for D^i}
Let $ a, b \in \mathbb{F} $. For every $ i \in \mathbb{N} $, it holds that
$$ \mathcal{D}_a^{i+1}(b) = \sigma \left( \mathcal{D}_a^i(b) \right) a + \delta \left( \mathcal{D}_a^i(b) \right). $$
\end{proposition}

We have that $ \mathcal{C}^{\sigma,\delta}_{L,k} \subseteq \mathbb{F}^n $ indeed is the linearized version of the code $ \mathcal{C}^{\sigma,\delta}_{\mathcal{B},k} \subseteq \mathbb{F}^\mathcal{B} $ from Definition \ref{def skew RS codes}. The proof follows from the definitions and Theorem \ref{th big commutative diagram}:

\begin{proposition} \label{prop indeed linearized version}
Let the notation be as in Theorem \ref{th big commutative diagram}, and fix $ k = 0, 1,2, \ldots, n $. If $ \mathcal{C}^{\sigma,\delta}_{\mathcal{B},k} \subseteq \mathbb{F}^\mathcal{B} $ and $ \mathcal{C}^{\sigma,\delta}_{L,k} \subseteq \mathbb{F}^n $ are the codes in Definition \ref{def skew RS codes} and Definition \ref{def linearized RS codes}, respectively, then it holds that
$$ \mathcal{C}^{\sigma,\delta}_{L,k} = \phi_\mathcal{B}(\mathcal{C}^{\sigma,\delta}_{\mathcal{B},k}). $$
\end{proposition}

To claim that linearized Reed-Solomon codes are maximum sum-rank distance, we now observe the corresponding Singleton bound, which follows from the fact that sum-rank weights are smaller than or equal to Hamming weights:

\begin{proposition} \label{prop singleton for decomposed rank}
Let the notation be as in Definition \ref{def sum rank metric}. For a left linear code $ \mathcal{C} \subseteq \mathbb{F}^n $ of dimension $ k $, if $ {\rm d_{SR}}(\mathcal{C}) $ denotes its minimum sum-rank distance, then it holds that
\begin{equation}
{\rm d_{SR}}(\mathcal{C}) \leq n - k + 1.
\label{eq singleton for sum rank}
\end{equation}
\end{proposition}

Thus we obtain the third main result of this section, which follows by combining Theorem \ref{th maximum skew distance codes}, Theorem \ref{th both weights coincide} and Proposition \ref{prop indeed linearized version}:

\begin{theorem} \label{th max sum rank distance}
Let the notation be as in Theorem \ref{th big commutative diagram}. If $ \mathcal{C}_{L,k} \subseteq \mathbb{F}^n $ is the linearized Reed-Solomon code from Definition \ref{def linearized RS codes}, then it holds that
$$ {\rm d_{SR}}(\mathcal{C}_{L,k}) = n - k + 1. $$
That is, the code $ \mathcal{C}_{L,k} $ is a maximum sum-rank distance code with lengths $ (n_1, n_2, \ldots, n_\ell) $ and division subrings $ (K_1, K_2, \ldots, K_\ell) $ (see Definition \ref{def sum rank metric}).
\end{theorem}

Collecting all the previous results, we also conclude the following:

\begin{theorem} \label{th range of parameters for maximum distance}
If the division subrings $ K_1, K_2, \ldots, K_\ell \subseteq \mathbb{F} $ are centralizers of pair-wise non-conjugate elements of $ \mathbb{F} $ for some skew polynomial ring $ \mathbb{F}[x;\sigma,\delta] $ (hence $ \ell $ is not larger than the number of conjugacy classes in $ \mathbb{F} $), and $ 1 \leq n_i \leq \dim_{K_i}(\mathbb{F}) $, for $ i =1,2, \ldots, \ell $, then there exists a $ k $-dimensional maximum sum-rank distance left linear code $ \mathcal{C} $ for all $ k = 1,2, \ldots, n $, for $ n = n_1 + n_2 + \cdots + n_\ell $, with lengths $ (n_1, n_2, \ldots, n_\ell) $ and division subrings $ (K_1, K_2, \ldots, K_\ell) $.
\end{theorem}

\begin{remark} \label{remark range of parameters}
If $ K = K_1 = K_2 = \ldots = K_\ell $, then any left linear code $ \mathcal{C} \subseteq \mathbb{F}^n $ with maximum rank distance over $ K $ is also maximum sum-rank distance for any vector of lengths $ (n_1, n_2, \ldots, n_\ell) $, such that $ n = n_1 + n_2 + \cdots + n_\ell $, and division subrings $ (K_1, K_2, \ldots, K_\ell) $. However, if $ m = \dim_K(\mathbb{F}) < \infty $, then such codes can only exist if $ n \leq m $. In such case, the previous theorem gives maximum sum-rank distance codes of any length $ n = 1,2, \ldots, \ell m $. See Subsection \ref{subsec finite fields} for a discussion when $ \mathbb{F} $ is a finite field.
\end{remark}

We conclude the subsection by comparing general linearized Reed-Solomon codes with codes from the literature. We start by the two main historical examples, that is, Reed-Solomon codes \cite{reed-solomon} and Gabidulin codes \cite{gabidulin, new-construction}:

\begin{example} [\textbf{Reed-Solomon codes and the Hamming metric}]
If $ \mathbb{F} $ is a field, the case of the Hamming metric and conventional Reed-Solomon codes \cite{reed-solomon} is recovered by choosing $ \sigma = {\rm Id} $ and $ \delta = 0 $. In such case, conjugacy classes are formed by one element, i.e. $ C(a) = \{a\} $, each centralizer satisfies $ K_a = \mathbb{F} $, and it holds that
$$ E_a(x^i) = \mathcal{D}_a^i(1) = a^i, $$
for all $ i \in \mathbb{N} $ and all $ a \in \mathbb{F} $. Hence skew Reed-Solomon codes and linearized Reed-Solomon codes coincide and give conventional Reed-Solomon codes. Moreover, the corresponding metric in $ \mathbb{F}^n $ is the sum-rank metric with lengths $ (1,1, \ldots,1) $ and subfields $ (\mathbb{F},\mathbb{F}, \ldots, \mathbb{F}) $ (both of size $ n $), which indeed is the Hamming metric. 

Finally, the restrictions in Theorem \ref{th range of parameters for maximum distance} for the possible parameters for which we obtain maximum distance codes are $ 1 \leq k \leq n \leq \# \mathbb{F} $, as in \cite{reed-solomon}.
\end{example}

\begin{example} [\textbf{Gabidulin codes and the rank metric}]
If $ \mathbb{F} $ is a field, the case of the rank metric and Gabidulin codes \cite{gabidulin, new-construction} can be recovered by choosing $ \sigma \neq {\rm Id} $, and then $ \delta = 0 $ and the conjugacy class $ C(1) $, or $ \delta $ an inner derivation ($ \delta = {\rm Id} - \sigma $) and the conjugacy class $ C(0) $. This second subfamily of linearized Reed-Solomon codes has been given in \cite[Section 4]{skew-evaluation1} when $ \mathbb{F} $ is finite.

We will focus on the first case (see Subsection \ref{subsec inner derivations} for the other case). It holds that $ \mathcal{D}_1 = \sigma $, $ K = K_1 = \mathbb{F}^\sigma $ is the subfield of elements of $ \mathbb{F} $ invariant by $ \sigma $, and $ K \subseteq \mathbb{F} $ is a field extension with Galois group $ G = \langle \sigma \rangle $ (finite or infinite). Then we have that 
$$ \mathcal{D}_1^i(a) = \sigma^i(a), $$
for all $ i \in \mathbb{N} $ and all $ a \in \mathbb{F} $. Hence the linearized notion of skew Reed-Solomon codes gives the Gabidulin codes from \cite{gabidulin, new-construction} when $ \mathbb{F} $ is finite. Moreover, the corresponding metric in $ \mathbb{F}^n $ is the sum-rank metric with lengths $ (n) $ and subfields $ (K) $, which indeed is the rank metric in $ \mathbb{F}^n $ over $ K $. 

Finally, the restrictions in Theorem \ref{th range of parameters for maximum distance} for the possible parameters for which we obtain maximum distance codes are $ 1 \leq k \leq n \leq \dim_K(\mathbb{F}) $, as in \cite{gabidulin, new-construction}.

Observe that this scenario has been studied for general fields in \cite{augot-function, augot, augot-extended} when $ \sigma $ has finite order, that is, $ K \subseteq \mathbb{F} $ is a finite extension of fields with cyclic Galois group. This example gives the general case: $ \sigma $ may have finite or infinite order.
\end{example}

On the other hand, the subfamily of $ k $-dimensional linearized Reed-Solomon codes when $ \mathbb{F} $ is finite, $ \delta = 0 $ and $ 1 \leq k \leq n_i $, for $ i = 1,2, \ldots, \ell $ (in particular, $ k \ell \leq n $), has been given in \cite[Section III-C]{skew-evaluation2} under the name \textit{pasting MDS construction}. Such pasting construction is also a linearized Reed-Solomon code over any field.

It is proven in \cite[Section III-C]{skew-evaluation2} that the pasting MDS construction gives an MDS linear code where each projection over $ \mathbb{F}^{n_i} $, for $ i = 1,2, \ldots, \ell $, gives an MRD code. However, these two properties separately do not imply having maximum sum-rank distance. We leave as open problem to see if such two properties combined are equivalent to having maximum sum-rank distance.

\section{Explicit descriptions over fields} \label{sec particular cases}

In this section, we study the above presented theory over fields. Throughout the whole section, $ \mathbb{F} $ is assumed to be a field, that is, commutative. We divide the study in two cases: $ \delta $ is an inner derivation or is not an inner derivation (see \cite[Section 8.3]{cohn} for a classification on general division rings). The first is always the case if $ \sigma \neq {\rm Id} $ and includes the case $ \delta = 0 $, and the second corresponds to standard non-zero derivations. The case of finite fields falls into the first category and is treated in a separate subsection. We also give an example of linearized Reed-Solomon codes only constructable by non-inner derivations.

\subsection{Inner derivations} \label{subsec inner derivations}

In this subsection, we study \textit{inner derivations} (see \cite[Section 8.3]{cohn} or \cite[Section 3]{lam-leroy}). We will show that these are the only derivations for non-identity endomorphisms and we show that such cases can always be reduced to zero derivations. However, as we will see in Subsection \ref{subsec only with non-inner der}, there are linearized Reed-Solomon codes over fields only constructable by non-inner derivations, hence the theory in our paper cannot always be simplified to using zero derivations. 

\begin{definition} [\textbf{Inner derivations \cite{cohn, lam-leroy}}]
Given an endomorphism $ \sigma : \mathbb{F} \longrightarrow \mathbb{F} $, we say that $ \delta : \mathbb{F} \longrightarrow \mathbb{F} $ is an inner $ \sigma $-derivation if there exists $ \gamma \in \mathbb{F} $ such that 
$$ \delta(b) = \gamma (b - \sigma(b)), $$
for all $ b \in \mathbb{F} $.
\end{definition}

It is a straightforward calculation to check that $ \delta $ indeed is a $ \sigma $-derivation. Next we observe that inner derivations are the only $ \sigma $-derivations over fields whenever $ \sigma \neq {\rm Id} $. This result is also proven in \cite[Section 8.3]{cohn}.

\begin{proposition} [\textbf{\cite{cohn}}] \label{prop every derivation is inner}
Let $ \sigma : \mathbb{F} \longrightarrow \mathbb{F} $ be an endomorphism and $ \delta : \mathbb{F} \longrightarrow \mathbb{F} $ be a $ \sigma $-derivation. If $ \sigma \neq {\rm Id} $, then $ \delta $ is an inner derivation.
\end{proposition}
\begin{proof}
Choose $ c \in \mathbb{F} $ such that $ \sigma (c) \neq c $. Since $ \mathbb{F} $ is commutative, it holds that
$$ \sigma(b)\delta(c) + \delta(b)c = \delta(bc) = \delta(cb) = \sigma(c)\delta(b) + \delta(c)b. $$
for all $ b \in \mathbb{F} $. Since $ c - \sigma(c) \neq 0 $, we deduce that
$$ \delta(b) = \left( \frac{\delta(c)}{c - \sigma(c)} \right) (b - \sigma(b)), $$
and we are done.
\end{proof}

As announced, the case of inner derivations can be reduced to that of zero derivations. The next result is given in \cite[Proposition 2.1.8]{liu-thesis} for finite fields, and is also proven in \cite[Section 8.3]{cohn}.

\begin{proposition}[\textbf{\cite{cohn, liu-thesis}}]
Let $ \sigma : \mathbb{F} \longrightarrow \mathbb{F} $ be an endomorphism, let $ \gamma \in \mathbb{F} $ and define the inner $ \sigma $-derivation $ \delta = \gamma ({\rm Id} - \sigma) $. The map
\begin{equation*}
\begin{array}{ccc}
\mathbb{F}[x; \sigma, \delta] & \longrightarrow & \mathbb{F}[x;\sigma, 0] \\
\sum_{i \in \mathbb{N}} F_i x^i & \mapsto & \sum_{i \in \mathbb{N}} F_i (x - \gamma)^i,
\end{array}
\end{equation*}
where $ F_i \in \mathbb{F} $, for all $ i \in \mathbb{N} $, is a ring isomorphism. 
\end{proposition}

As expected, we have the same descriptions of centralizers and conjugacy classes for all inner derivations. The proof is straightforward. 

\begin{proposition} \label{prop properties inner derivations}
Let $ \sigma : \mathbb{F} \longrightarrow \mathbb{F} $ be an endomorphism, let $ \gamma \in \mathbb{F} $ and define the inner $ \sigma $-derivation $ \delta = \gamma ({\rm Id} - \sigma) $. For $ a \in \mathbb{F} $, the following hold:
\begin{enumerate}
\item
$ K = K_a = \mathbb{F}^\sigma = \{ b \in \mathbb{F} \mid \sigma(b) = b \} $ if $ a \neq \gamma $, and $ K_\gamma = \mathbb{F} $.
\item
$ C(a) = \{ c \in \mathbb{F} \mid c-\gamma = (a-\gamma)\sigma(b)b^{-1}, b \in \mathbb{F}^* \} $ if $ a \neq \gamma $, and $ C(\gamma) = \{ \gamma \} $. 
\end{enumerate}
\end{proposition}

We may simplify the calculation of $ \mathcal{D}_a^i(b) $ from Proposition \ref{prop iterative formula for D^i} when $ \delta = 0 $ as follows:

\begin{proposition}
If $ \delta = 0 $, then for all $ a,b \in \mathbb{F} $ and all $ i \in \mathbb{N} $, it holds that
$$ \mathcal{D}_a^{i+1}(b) = \sigma (\mathcal{D}_a^i(b))a \quad \textrm{and} \quad \mathcal{D}_a^i(b) = \sigma^i(b) N_i(a), $$
where $ N_i(a) = a \sigma(a) \cdots \sigma^{i-1}(a) $.
\end{proposition}
 
In \cite[Section 4]{matroidal}, the equivalence between the matroid of P-independent subsets of one conjugacy class and linearly independent sets of points in the projective space $ \mathbb{P}_K(\mathbb{F}) $ is given for finite fields. We give a similar result in the language of lattices. The proof is analogous to the one that we will give for Proposition \ref{prop projective description non-inner derivations}.

\begin{proposition}  \label{prop projective description inner derivations}
Fix $ a \in \mathbb{F} $ and define in $ C(a) $ the sum of two P-closed sets $ \Omega_1, \Omega_2 \subseteq C(a) $ as $ \Omega_1 + \Omega_2 = \overline{\Omega_1 \cup \Omega_2} \subseteq C(a) $. The collection of P-closed subsets of $ C(a) $ forms a lattice with sums and intersections isomorphic to the lattice of projective subspaces of $ \mathbb{P}_K(\mathbb{F}) $ by the bijection
\begin{equation*}
\begin{array}{ccc}
\mathbb{P}_K(\mathbb{F}) & \longrightarrow & C(a) \\
 \left[ b \right] & \mapsto & (a-\gamma)\sigma(b)b^{-1} + \gamma,
\end{array}
\end{equation*}
where $ \left[ b \right] = \{ \lambda b \mid \lambda \in K^* \} \in \mathbb{P}_K(\mathbb{F}) $.
\end{proposition}

\subsection{Finite fields} \label{subsec finite fields}

We discuss in this subsection the case where $ \mathbb{F} = \mathbb{F}_{p^s} $ is the finite field with $ p^s $ elements, for a positive $ s \in \mathbb{N} $ and a prime $ p $. The description in this case of conjugacy classes is due to \cite{matroidal}. Constructions and Hamming-metric decoding algorithms of skew Reed-Solomon codes in this case are given in \cite{skew-evaluation1, skew-evaluation2}. However, the definitions of skew metrics and linearized Reed-Solomon codes, their connection with skew Reed-Solomon codes, the fact that these are maximum skew distance and linearized Reed-Solomon codes are maximum sum-rank distance are all new even in this case.

Let $ \sigma : \mathbb{F}_{p^s} \longrightarrow \mathbb{F}_{p^s} $ be given by $ \sigma(a) = a^{p^r} $, for $ a \in \mathbb{F} $ and $ r \in \mathbb{N} $. Then $ \mathbb{F}^\sigma = \mathbb{F}_q $, where $ q = p^d $, $ d = {\rm gcd}(r,s) $. Let $ m \in \mathbb{N} $ be such that $ q^m = p^s $. We start by the following result:

\begin{proposition} 
Every $ \sigma $-derivation in $ \mathbb{F}_{q^m} $ is an inner derivation. In particular, $ \delta = 0 $ is the only $ {\rm Id} $-derivation in $ \mathbb{F}_{q^m} $.
\end{proposition}
\begin{proof}
We have proven it in Proposition \ref{prop every derivation is inner} for $ \sigma \neq {\rm Id} $ and any field. If $ \sigma = {\rm Id} $ and $ a \in \mathbb{F}_{q^m} $, it follows from
$$ \delta(a) = \delta(a^{q^m}) = \sum_{i=0}^{q^m - 1} a^i a^{q^m - i - 1} \delta(a) = (a-a)^{q^m - 1} \delta(a) = 0. $$
\end{proof}

Therefore the study in the previous subsection applies directly to finite fields. Fix now $ \gamma \in \mathbb{F}_{q^m} $ and define $ \delta = \gamma ({\rm Id} - \sigma) $. The following result is \cite[Corollary 1]{matroidal}:

\begin{proposition} [\textbf{\cite{matroidal}}]
If $ a \in \mathbb{F}_{q^m} $ and $ a \neq \gamma $, then 
$$ K_a = \mathbb{F}_{q^m}^\sigma = \mathbb{F}_q \quad \textrm{and} \quad \# C(a) = \# \mathbb{P}_{\mathbb{F}_q}(\mathbb{F}_{q^m}) = \frac{q^m - 1}{q-1}. $$ 
In particular, there are $ q - 1 $ non-trivial conjugacy classes in $ \mathbb{F}_{q^m} $.
\end{proposition}

We conclude by comparing the ranges of parameters for which we may construct Gabidulin codes \cite{gabidulin, new-construction} and general linearized Reed-Solomon codes in this context using the extension $ \mathbb{F}_q \subsetneqq \mathbb{F}_{q^m} $ (See also Theorem \ref{th range of parameters for maximum distance} and Remark \ref{remark range of parameters}). Both may have any dimension, so we compare only their maximum lengths.

The maximum length of a Gabidulin code in $ \mathbb{F}_{q^m} $ over the base field $ \mathbb{F}_q $ is $ n_G = m $, whereas the maximum length of a linearized Reed-Solomon code in such case is $ n_L = (q-1)m $. That is,
$$ n_L = (q-1)n_G. $$
Conversely, fix a length $ n $. The minimum extension degree $ m $ for a Gabidulin code is $ m_G = n $, whereas for a linearized Reed-Solomon code it is $ m_L = n/(q-1) $. If $ q_G = q^{m_G} $ and $ q_L = q^{m_L} $, then
$$ q_G = q_L^{q-1}. $$

\subsection{Non-inner derivations} \label{subsec only with non-inner der}

In this subsection, we study non-inner derivations when $ \sigma = {\rm Id} $. We will give analogous descriptions as those in Subsection \ref{subsec inner derivations}. We conclude by giving an example of linearized Reed-Solomon codes that can only be constructed by using non-inner derivations. This shows that we may not reduce the study of linearized Reed-Solomon codes to the case of zero derivations when considering infinite fields. 

\begin{proposition} 
Let $ \delta : \mathbb{F} \longrightarrow \mathbb{F} $ be an $ {\rm Id} $-derivation. For $ a \in \mathbb{F} $, the following hold:
\begin{enumerate}
\item
$ K = K_a = \mathbb{F}^\delta = \{ b \in \mathbb{F} \mid \delta(b) = 0 \} $.
\item
$ C(a) = \{ a + \delta(b)b^{-1} \in \mathbb{F} \mid b \in \mathbb{F}^* \} $. 
\end{enumerate}
\end{proposition}

\begin{proposition}   \label{prop projective description non-inner derivations}
Fix $ a \in \mathbb{F} $ and let the notation be as in Proposition \ref{prop projective description inner derivations}. The collection of P-closed subsets of $ C(a) $ forms a lattice with sums and intersections isomorphic to the lattice of projective subspaces of $ \mathbb{P}_K(\mathbb{F}) $ by the bijection
\begin{equation*}
\begin{array}{ccc}
\mathbb{P}_K(\mathbb{F}) & \longrightarrow & C(a) \\
 \left[ b \right] & \mapsto & a + \delta(b)b^{-1}.
\end{array}
\end{equation*}
\end{proposition} 
\begin{proof}
It is easy to see that the given map is well-defined and onto. We now prove that it is one to one. Assume that $ \delta(b)b^{-1} = \delta(c)c^{-1} $, for some $ b,c \in \mathbb{F}^* $. Then
$$ \delta(bc^{-1}) = \delta(b)c^{-1} - b\delta(c)c^{-2} = (\delta(b)c - b \delta(c)) c^{-2} = 0. $$
Thus $ bc^{-1} \in K^* $, hence $ [b] = [c] $ and the map is bijective. Finally, it follows directly from Lemma \ref{lemma linearized in one conj class} and Corollary \ref{cor linearized in one conj class} that the P-closed subsets of $ C(a) $ form a lattice with the given sums and intersections that is isomorphic to $ \mathbb{P}_K(\mathbb{F}) $ by the given bijection.
\end{proof}

We conclude with the above mentioned example of linearized Reed-Solomon codes only constructable by non-inner derivations:

\begin{example}
Let $ \mathbb{F} = \mathbb{F}_p(z) $ be the field of rational functions over $ \mathbb{F}_p $, where $ p $ is a prime. Consider $ K = \mathbb{F}_p(z^p) $. First, we have that the only endomorphism $ \sigma : \mathbb{F} \longrightarrow \mathbb{F} $ such that $ \sigma(a) = a $, for all $ a \in K $, is the identity endomorphism: Let $ \sigma $ be such an endomorphism. We have that
$$ \sigma \left( \sum_{i = 0}^d a_i z^i \right) = \sum_{i = 0}^d a_i \sigma(z)^i, $$
for all $ a_0, a_1, \ldots, a_d \in \mathbb{F}_p $, and for all $ d \in \mathbb{N} $. Let $ f(z) = \sigma(z) \in \mathbb{F}_p(z) $. Since $ \sigma(z^p) = z^p $ by hypothesis, it holds that
$$ z^p = \sigma(z^p) = \sigma(z)^p = f(z)^p = f(z^p). $$
Therefore $ f(z) = z $ and $ \sigma = {\rm Id} $.

This means that we may not use either zero or non-zero inner derivations to construct linearized Reed-Solomon codes using $ \mathbb{F} $ with centralizers $ K $. However, if $ \delta : \mathbb{F} \longrightarrow \mathbb{F} $ is the standard derivation
$$ \delta(f(z)) = \frac{d}{dz}(f(z)), $$
for all $ f(z) \in \mathbb{F} $, then $ K_a = \mathbb{F}_p(z^p) $, for all $ a \in \mathbb{F} $. Hence we conclude that we may construct linearized Reed-Solomon codes where all centralizers are $ K = \mathbb{F}_p(z^p) $ using the non-inner derivation $ \delta $. In particular, we may construct Gabidulin-type codes which are maximum rank distance for the field extension $ \mathbb{F}_p(z^p) \subsetneqq \mathbb{F}_p(z) $ using $ \delta $.
\end{example}

\section*{Acknowledgement}

The author gratefully acknowledges the support from The Danish Council for Independent Research (Grant No. DFF-7027-00053B and Grant No. DFF-5137-00076B, ``EliteForsk-Rejsestipendium'').

\appendix

\section{Alternative proof of Lemma \ref{lemma evaluation of skew as operator}} \label{app proof of th evaluation}

In this appendix, we give an alternative short proof of Lemma \ref{lemma evaluation of skew as operator}. Fix $ a \in \mathbb{F} $ and denote $ \mathcal{D} = \mathcal{D}_a $. Fix $ b \in \mathbb{F}^* $ and denote by $ N_i(b) $ the evaluation of $ x^i \in \mathbb{F}[x;\sigma,\delta] $ in $ b $, for all $ i \in \mathbb{N} $. It follows from \cite[Lemma 2.4]{lam-leroy} and \cite[Eq. (2.3)]{lam-leroy} that
\begin{equation}
N_{i+1}(b) = \sigma(N_i(b)) b + \delta(N_i(b)),
\label{eq formula for N_i}
\end{equation}
for all $ i \in \mathbb{N} $. By linearity, we only need to prove, for all $ i \in \mathbb{N} $, that
$$ N_i(\mathcal{D}(b)b^{-1}) = \mathcal{D}^i(b)b^{-1}. $$
By (\ref{eq formula for N_i}), we only need to prove, for all $ i \in \mathbb{N} $, that
\begin{equation}
\mathcal{D}^{i+1}(b)b^{-1} = \sigma(\mathcal{D}^i(b)b^{-1}) \mathcal{D}(b)b^{-1} + \delta (\mathcal{D}^i(b)b^{-1}).
\label{eq theorem for D^i}
\end{equation}
Fix $ i \in \mathbb{N} $. Expanding $ \mathcal{D}(b) $ and $ \delta(\mathcal{D}^i(b)b^{-1}) $ on the right-hand side of (\ref{eq theorem for D^i}), we see that
$$ \sigma(\mathcal{D}^i(b)b^{-1}) \mathcal{D}(b) + \delta (\mathcal{D}^i(b)b^{-1})b = \sigma (\mathcal{D}^i(b))a + \delta(\mathcal{D}^i(b)) = \mathcal{D}^{i+1}(b), $$
and we are done.

\bibliographystyle{plain}

\end{document}